\documentclass[12pt]{amsart}
\usepackage{}
\usepackage{amsmath}
\usepackage{amsfonts}
\usepackage{amssymb}
\usepackage[all]{xy}           
\usepackage{bm}
\usepackage{bbding}
\usepackage{txfonts}
\usepackage{amscd}
\usepackage{xspace}
\usepackage[shortlabels]{enumitem}
\usepackage{ifpdf}

\ifpdf
  \usepackage[colorlinks,final,backref=page,hyperindex]{hyperref}
\else
  \usepackage[colorlinks,final,backref=page,hyperindex,hypertex]{hyperref}
\fi
\usepackage{tikz}
\usepackage[active]{srcltx}

\makeatletter

\newtheorem{thm}{Theorem}[section]
\newtheorem{lem}[thm]{Lemma}
\newtheorem{cor}[thm]{Corollary}
\newtheorem{pro}[thm]{Proposition}

\theoremstyle{definition}
\newtheorem{rmk}[thm]{Remark}
\newtheorem{defi}[thm]{Definition}

\newcommand{\nc}{\newcommand}
\newcommand{\delete}[1]{}

\nc{\mlabel}[1]{\label{#1}}  
\nc{\mcite}[1]{\cite{#1}}  
\nc{\mref}[1]{\ref{#1}}  
\nc{\mbibitem}[1]{\bibitem{#1}} 

\delete{
\nc{\mlabel}[1]{\label{#1}{\hfill \hspace{1cm}{\bf{{\ }\hfill(#1)}}}}
\nc{\mcite}[1]{\cite{#1}{{\em{{\ }(#1)}}}}  
\nc{\mref}[1]{\ref{#1}{{\em{{\ }(#1)}}}}  
\nc{\mbibitem}[1]{\bibitem[\em #1]{#1}} 
}

\newcommand {\emptycomment}[1]{}

\nc{\oprn}{\theta}
\nc{\Oprn}{\Theta}

\nc{\calo}{\mathcal{O}}
\nc{\oop}{$\mathcal{O}$-operator\xspace}
\nc{\oops}{$\mathcal{O}$-operators\xspace}
\nc{\mrho}{{\bm{\varrho}}}
\nc{\emk}{\mathbf{K}}
\nc{\invlim}{\displaystyle{\lim_{\longleftarrow}}\,}
\nc{\ot}{\otimes}

\newcommand{\lon }{\,\rightarrow\,}
\newcommand{\be }{\begin{equation}}
\newcommand{\ee }{\end{equation}}

\newcommand{\LR}{$\mathsf{Lie}\mathsf{Rep}$~}

\newcommand{\MC}{$\mathsf{MC}$~}

\newcommand{\RB}{$\mathsf{RB}$~}
\newcommand{\g}{\mathfrak g}
\newcommand{\h}{\mathfrak h}

\newcommand{\huaD}{\mathcal{D}}

\newcommand{\huaH}{\mathcal{H}}

\newcommand{\huaO}{{\mathcal{O}}}

\newcommand{\frki}{\mathfrak i}

\newcommand{\frkl}{\mathfrak l}

\newcommand{\frkp}{\mathfrak p}

\newcommand{\frkC}{\mathfrak C}

\newcommand{\Courant}[1]{\left\llbracket  #1\right\rrbracket }

\newcommand{\Id}{{\rm{Id}}}

\newcommand{\br}[1]{   [ \cdot,    \cdot  ]   }

\newcommand{\dM}{\mathrm{d}}

\newcommand{\Hom}{\mathrm{Hom}}
\newcommand{\Der}{\mathrm{Der}}
\newcommand{\Lie}{\mathrm{Lie}}

\newcommand{\TLB}{\mathrm{TLB}}
\newcommand{\CE}{\mathsf{CE}}

\newcommand{\SN}{\mathsf{SN}}

\newcommand{\NR}{\mathsf{NR}}

\newcommand{\gl}{\mathfrak {gl}}

\newcommand{\ad}{\mathrm{ad}}

\newcommand{\Img}{\mathrm{Im}}

\newcommand{\Sym}{\mathsf{S}}
\newcommand{\Ten}{\mathsf{T}}

\nc{\CV}{\mathbf{C}}

\begin{document}

\title[A survey on deformations, cohomologies  and homotopies  of relative RB Lie algebras]{A survey   on deformations, cohomologies  and homotopies  of relative Rota-Baxter Lie algebras}

\author{Yunhe Sheng}
\address{Department of Mathematics, Jilin University, Changchun 130012, Jilin, China, and Sino-Russian Mathematics Center, Peking University, Beijing 100871, China.}
\email{shengyh@jlu.edu.cn}


\begin{abstract}
In this paper, we review deformation, cohomology and homotopy theories of relative Rota-Baxter ($\mathsf{RB}$) Lie algebras, which have attracted quite much interest  recently. Using Voronov's higher derived brackets, one can obtain an $L_\infty$-algebra whose Maurer-Cartan elements are relative \RB Lie algebras. Then using the twisting method, one can obtain the $L_\infty$-algebra that controls deformations of a relative \RB Lie algebra. Meanwhile, the cohomologies of relative \RB Lie algebras can also be defined with the help of the twisted $L_\infty$-algebra.
Using the controlling algebra approach, one can also introduce the notion of   homotopy relative \RB Lie algebras with close connection to pre-Lie$_\infty$-algebras. Finally, we briefly review deformation, cohomology and homotopy theories of relative $\mathsf{RB}$ Lie algebras of nonzero weights.

\end{abstract}


\keywords{Cohomology, deformation, homotopy, $L_\infty$-algebra,   Rota-Baxter algebra,    triangular Lie bialgebra}

\maketitle

\vspace{-1.1cm}

\tableofcontents

\allowdisplaybreaks

\section{Introduction}\mlabel{sec:intr}

The concept of Rota-Baxter ($\mathsf{RB}$) operators on associative algebras was
introduced   by G. Baxter \cite{Ba} in his study of
fluctuation theory in probability. Recently it has found many
applications, including Connes-Kreimer's~\cite{CK} algebraic
approach to the renormalization in perturbative quantum field
theory. \RB operators lead to the splitting of
operads~\cite{BBGN,PBG}, and are closely related to quasisymmetric functions and Hopf algebras \cite{Fard,Yu-Guo}.  Recently the relationship between \RB operators and double Poisson algebras were studied in \cite{Goncharov}.
In the Lie algebra context, a \RB operator was introduced independently in the 1980s as the
operator form of the classical Yang-Baxter equation. For further details on
\RB operators, see ~\cite{Gub-AMS,Gub}. To better understand the classical Yang-Baxter equation and
 related integrable systems, the more general notion of a relative \RB operator (which was called an \oop in the original literature)
on a Lie algebra was introduced by Kupershmidt~\cite{Ku}. Relative \RB operators provide solutions of the classical Yang-Baxter equation in the semidirect product Lie algebra and give rise to pre-Lie algebras \cite{Bai}.

The concept of a formal deformation of an algebraic structure began with the seminal
work of Gerstenhaber~\cite{Ge0,Ge} for associative
algebras. Nijenhuis and Richardson   extended this study to Lie algebras
~\cite{NR,NR2}. See \cite{GLST1} for more details about the deformation theories of various algebraic structures. More generally, deformation theory
for algebras over quadratic operads was developed by Balavoine~\cite{Bal}. For more
general operads we refer the reader to \cite{KSo,LV,Ma}, and the references therein. 
There is a well known slogan, often attributed to Deligne, Drinfeld and Kontsevich:  {\em every reasonable deformation theory is controlled by a differential graded   Lie algebra, determined up to quasi-isomorphism}. This slogan has been made into a rigorous theorem by Lurie  and Pridham, cf. \cite{Lu,Pr}, and a recent simple treatment in  \cite{GLST}.
It is also meaningful to deform {\em maps} compatible with given algebraic structures. Recently, the deformation theory of morphisms was   developed in \cite{Borisov,Fregier-Zambon-1,Fregier-Zambon-2} and the deformation theory of diagrams of algebras was studied in \cite{Barmeier,Fregier} using the minimal model of operads and the method of derived brackets \cite{Kosmann-Schwarzbach,Ma-0,Vo}.
Sometimes a differential graded  Lie algebra up to quasi-isomorphism  controlling a deformation theory manifests itself naturally as an \emph{$L_\infty$-algebra}. This often happens when one tries to deform several algebraic structures as well as a compatibility relation between them, such as  diagrams of algebras mentioned above.

A classical approach for studying a mathematical structure is associating invariants to it. Prominent among these are cohomological invariants, or simply cohomology, of various types of algebras. Cohomology controls deformations and extension problems of the corresponding algebraic structures. Cohomology theories of various kinds of algebras have been developed and studied in \cite{Ch-Ei,Ge0,Har,Hor}. More recently these classical constructions have been extended to strong homotopy (or infinity) versions of the algebras, cf. for example \cite{Hamilton-Lazarev}.

Homotopy invariant algebraic structures play a prominent role in modern mathematical physics.
Historically, the first such structure was that of an $A_\infty$-algebra introduced by Stasheff
in his study of based loop spaces~\cite{Sta63}.
Relevant later  developments include the work of   Lada and Stasheff~\cite{LS,stasheff:shla} about
$L_\infty$-algebras in mathematical physics
and the work of Chapoton and Livernet~\cite{CL} about
pre-Lie$_\infty$-algebras. Strong homotopy (or infinity-) versions of a large class of algebraic
structures were studied in the context of
operads in ~\cite{LV, MSS}.

Due to the importance of relative \RB Lie and associative algebras, the studies of corresponding deformation, cohomology and homotopy theories attract much interest recently. The first step toward such a study was given in \cite{TBGS-1}, where the deformation  and cohomology theories of relative \RB operators were established and applications were given to study deformations and cohomologies of skew-symmetric $r$-matrices. See also the survey article \cite{TBGS-s} for more details. Then in \cite{LST}, applying Voronov's higher derived brackets \cite{Vo}, the controlling algebra of relative \RB Lie algebras (a relative \RB Lie algebra consists of a Lie algebra $\g$, a representation of $\g$ on a vector space $V$ and a relative \RB operator $T:V\to\g$) was constructed, which turns out to be an $L_\infty$-algebra. Then using the twisting method via Maurer-Cartan elements given in \cite{Getzler}, one obtain a twisted $L_\infty$-algebra that governs simultaneous deformations of relative \RB Lie algebras. Using the $l_1$ in the twisted $L_\infty$-algebra, one can define the cohomology of relative \RB Lie algebras. Finally, one can define homotopy relative \RB operators via Maurer-Cartan characterization.  Voronov's higher derived brackets and the controlling algebras of homotopy relative \RB Lie algebras were studied more intrinsically in \cite{LST2} via the functorial approach. Note that the aforementioned relative \RB operators are of weight 0, and deformation, cohomology and homotopy theories of relative \RB operators and relative \RB Lie algebras of nonzero weights were further studied in \cite{CC,Das2108,Das2109,JSZ,TBGS-2}.

In the associative algebra context, deformations, cohomologies and homotopies of relative \RB associative algebras of weight 0 were studied in \cite{Das,DasM}. Deformations and cohomologies of relative \RB operators of nonzero weights were studied in \cite{Das2108}. Independently, deformations, cohomologies and homotopies of   \RB associative algebras of nonzero weights were studied in \cite{WZ}. In particular, it was shown in \cite{WZ} that  the operad governing homotopy \RB associative  algebras is a minimal model of the operad of \RB associative algebras. Note that due to the nonhomogeneous
relations, the operad of \RB algebras are not quadratic, and not covered by the Koszul duality theory. In \cite{DK}, Dotsenko and Khoroshkin gave a detailed study of the operad of  \RB associative algebras, and note that it is very difficult to give explicit formulas for differentials in the free resolutions. So it is still curious to give the homotopy theory of \RB algebras using the purely operadic approach.

 The paper is organized as follows. In Section \ref{sec:pre}, we recall the main tools which will be used frequently: the  Nijenhuis-Richardson bracket and higher derived brackets. In Section \ref{sec:def}, we survey the deformation theory of relative \RB Lie algebras. Given vector spaces $\g$ and $V$, first using Voronov's higher derived brackets  one obtains an $L_\infty$-algebra, whose Maurer-Cartan elements are relative \RB Lie algebra structures on $\g$ and $V$. Then given a relative \RB Lie algebra, applying the twisting theory via Maurer-Cartan elements, one obtains a twisted $L_\infty$-algebra governs deformations of the relative \RB Lie algebra.
   In Section \ref{sec:coh}, we survey the cohomology theory of relative \RB Lie algebras. Using the $l_1$ in the above twisted $L_\infty$-algebra, one can define the cohomology of relative \RB Lie algebras.
   Moreover,   there is a long exact sequence of cohomology groups linking the cohomology of \LR pairs introduced in \cite{Arnal}, the cohomology of $\huaO$-operators introduced in \cite{TBGS-1} and the cohomology of relative \RB Lie algebras. The above general framework has two important special cases: \RB Lie algebras and triangular Lie bialgebras. In Section \ref{sec:cohomologyRB}, one can apply the above general framework to introduce the cohomology of \RB Lie algebras. In Section \ref{sec:cohomologyTLB}, one can apply the above general framework to introduce the cohomology of triangular  Lie bialgebras.
In Section \ref{sec:hom}, we survey  homotopy relative \RB Lie algebras that obtained through the Maurer-Cartan approach. In Section \ref{sec:ass}, we briefly survey deformation, cohomology and homotopy theories of \RB Lie and associative algebras of nonzero weights.

\section{The Nijenhuis-Richardson bracket and higher derived brackets}\label{sec:pre}

In this section, we recall the Nijenhuis-Richardson bracket and higher derived brackets which are the main tools in  later sections.

\subsection{The Nijenhuis-Richardson bracket}

Let $\g$ be a vector space. For all $n\ge 0$, set $C^n(\g,\g):=\Hom(\wedge^{n+1}\g,\g).$ Consider the graded vector space
$
C^*(\g,\g)=\oplus_{n=0}^{+\infty}C^{n}(\mathfrak{g},\mathfrak{g})=\oplus_{n=0}^{+\infty}\Hom(\wedge^{n+1}\g,\g).
$
Then $C^*(\g,\g)$ equipped with the {\em Nijenhuis-Richardson bracket} \cite{NR,NR2}
\begin{eqnarray}\label{NR-bracket}
[P,Q]_{\NR}=P\bar{\circ}Q-(-1)^{pq}Q\bar{\circ} P,\,\,\,\,\forall P\in C^{p}(\mathfrak{g},\mathfrak{g}),
Q\in C^{q}(\mathfrak{g},\mathfrak{g}),
\end{eqnarray}
is a graded Lie algebra, where  $P\bar{\circ}Q\in C^{p+q}(\mathfrak{g},\mathfrak{g})$ is defined by
\begin{eqnarray}\label{NR-bracket-com}
(P\bar{\circ}Q)(x_{1},\cdots,x_{p+q+1})=\sum_{\sigma\in \mathbb S_{(q+1,p)}}(-1)^{\sigma}
P(Q(x_{\sigma(1)},\cdots,x_{\sigma(q+1)}),
x_{\sigma(q+2)},\cdots,x_{\sigma(p+q+1)}).
\end{eqnarray}
Here $\mathbb S_{(i,n-i)}$ denote the set of $(i,n-i)$-shuffles. Recall that a permutation $\sigma\in\mathbb S_n$ is called an {\em $(i,n-i)$-shuffle} if $\sigma(1)<\cdots <\sigma(i)$ and $\sigma(i+1)<\cdots <\sigma(n)$. If $i=0$ or $n$, we assume $\sigma=\Id$.  The notion of an $(i_1,\cdots,i_k)$-shuffle and the set $\mathbb S_{(i_1,\cdots,i_k)}$ are defined analogously.
 For $\mu\in C^{1}(\g,\g)=\Hom(\wedge^2\g,\g)$, we have
\begin{eqnarray*}
[\mu,\mu]_{\NR}(x,y,z)=2(\mu\bar{\circ}\mu)(x,y,z)
                      =2\Big(\mu(\mu(x,y),z)+\mu(\mu(y,z),x)+\mu(\mu(z,x),y)\Big).
\end{eqnarray*}
Thus, $\mu$ defines a Lie algebra structure on $\g$ if and only if $[\mu,\mu]_{\NR}=0$.

Let $(\g,\mu)$ be a Lie algebra. Define the set of $0$-cochains $\frkC^0_\Lie(\g;\g)$ to be $0$, and define the set of $n$-cochains $\frkC^n_\Lie(\g;\g)$ to be
$$
\frkC^n_\Lie(\g;\g):=\Hom(\wedge^n\g,\g)=C^{n-1}(\g,\g),\quad n\geq 1.
$$
The {\em Chevalley-Eilenberg coboundary operator} $\dM_\CE$   of the Lie algebra $\g$ with coefficients in the adjoint representation  is defined by
\begin{eqnarray}\label{CE-cohomology-diff}
\dM_\CE f=(-1)^{n-1}[\mu,f]_{\NR},\quad \forall f\in \frkC^n_\Lie(\g;\g).
\end{eqnarray}
The resulting cohomology is denoted by $\huaH^*_{\Lie}(\g;\g)$.

Let $\g_1$ and $\g_2$ be two vector spaces and elements in $\g_1$ will be denoted by $x,y,z, x_i$ and elements in $\g_2$ will be denoted by $u,v,w,v_i$. 
For a   multilinear map $f:\wedge^{k}\g_1\otimes\wedge^{l}\g_2\lon\g_1$, we define  $\hat{f}\in C^{k+l-1}\big(\g_1\oplus\g_2,\g_1\oplus\g_2\big)$ by
\begin{eqnarray*}
\hat{f}\big((x_1,v_1),\cdots,(x_{k+l},v_{k+l})\big):=\sum_{\tau\in\mathbb S_{(k,l)}}(-1)^{\tau}\Big(f(x_{\tau(1)},\cdots,x_{\tau(k)},v_{\tau(k+1)},\cdots,v_{\tau(k+l)}),0\Big).
\end{eqnarray*}
Similarly, for $f:\wedge^{k}\g_1\otimes\wedge^{l}\g_2\lon\g_2$, we define   $\hat{f}\in C^{k+l-1}\big(\g_1\oplus\g_2,\g_1\oplus\g_2\big)$ by
\begin{eqnarray*}
\hat{f}\big((x_1,v_1),\cdots,(x_{k+l},v_{k+l})\big):=\sum_{\tau\in\mathbb S_{(k,l)}}(-1)^{\tau}\Big(0,f(x_{\tau(1)},\cdots,x_{\tau(k)},v_{\tau(k+1)},\cdots,v_{\tau(k+l)})\Big).
\end{eqnarray*}
The linear map $\hat{f}$ is called a {\em lift} of $f$.
Define $\g^{k,l}:=\wedge^k\g_1\otimes\wedge^{l}\g_2$.
The vector space $\wedge^{n}(\g_1\oplus\g_2)$ is isomorphic to the direct sum of $\g^{k,l},~k+l=n$.

\begin{defi}
A linear map $f\in \Hom\big(\wedge^{k+l+1}(\g_1\oplus\g_2),\g_1\oplus\g_2\big)$ has a {\em bidegree} $k|l$, which is denoted by $||f||=k|l$,   if   $f$ satisfies the following two conditions:
\begin{itemize}
\item[\rm(i)] If $X\in\g^{k+1,l}$, then $f(X)\in\g_1$ and if $X\in\g^{k,l+1}$, then $f(X)\in\g_2;$
\item[\rm(ii)]
  In all the other cases $f(X)=0.$
\end{itemize}
We denote the set of homogeneous linear maps of bidegree $k|l$ by $C^{k|l}(\g_1\oplus\g_2,\g_1\oplus\g_2)$.
\end{defi}

It is clear that this gives a well-defined bigrading on the vector space $\Hom\big(\wedge^{k+l+1}(\g_1\oplus\g_2),\g_1\oplus\g_2\big)$.
We have $k+l\ge0,~k,l\ge-1$ because $k+l+1\ge1$ and $k+1,~l+1\ge0$.

The following lemmas are very important in  later studies.

\begin{lem}\label{important-lemma-2}
The Nijenhuis-Richardson bracket on $C^*(\g_1\oplus\g_2,\g_1\oplus\g_2)$ is compatible with the  bigrading. More precisely, if $||f||=k_f|l_f$, $||g||=k_g|l_g$, then $||[f,g]_{\NR}||=(k_f+k_g)|(l_f+l_g).$
\end{lem}
\begin{proof}
It follows from direct computation.
\end{proof}

 \begin{rmk}
  In   later studies, the subspaces $C^{k|0}(\g_1\oplus\g_2,\g_1\oplus\g_2)$ and $C^{-1|l}(\g_1\oplus\g_2,\g_1\oplus\g_2)$ will be frequently used. By the above lift map, one has the following isomorphisms:
  \begin{eqnarray}
   \label{rep-cochain-1} C^{k|0}(\g_1\oplus\g_2,\g_1\oplus\g_2)&\cong & \Hom(\wedge^{k+1}\g_1,\g_1)\oplus \Hom(\wedge^{k}\g_1\otimes \g_2,\g_2),\\
    \label{rep-cochain-2}C^{-1|l}(\g_1\oplus\g_2,\g_1\oplus\g_2)&\cong & \Hom(\wedge^{l}\g_2,\g_1).
  \end{eqnarray}
\end{rmk}

\begin{lem}\label{Zero-condition-2}
If $||f||=(-1)|k$ and $||g||=(-1)|l$, then $[f,g]_{\NR}=0.$ Consequently, $\oplus_{l=1}^{+\infty}  C^{-1|l}(\g_1\oplus\g_2,\g_1\oplus\g_2)$ is an abelian subalgebra of the graded Lie algebra $(C^*(\g_1\oplus \g_2,\g_1\oplus\g_2),[\cdot,\cdot]_\NR)$
\end{lem}
\begin{proof}
  It follows from Lemma \ref{important-lemma-2}.
\end{proof}

\subsection{$L_\infty$-algebras and higher derived brackets}

The notion of an $L_\infty$-algebra was introduced by Stasheff in \cite{stasheff:shla}. See  \cite{LS,LM} for more details.

Let $V=\oplus_{k\in\mathbb Z}V^k$ be a $\mathbb Z$-graded vector space. 
We will denote by $\Sym(V)$ the {\em symmetric  algebra}  of $V$. That is,
$
\Sym(V):=\Ten(V)/I,
$
where $\Ten(V)$ is the tensor algebra and $I$ is the $2$-sided ideal of $\Ten(V)$ generated by all homogeneous elements  of the form
$
x\otimes y-(-1)^{xy}y\otimes x.
$
We will write $\Sym(V)=\oplus_{i=0}^{+\infty}\Sym^i (V)$. 
Moreover, we denote the reduced symmetric  algebra by $\bar{\Sym}(V):=\oplus_{i=1}^{+\infty}\Sym^{i}(V)$. Denote the product of homogeneous elements $v_1,\cdots,v_n\in V$ in $\Sym^n(V)$ by $v_1\odot\cdots\odot v_n$. The degree of $v_1\odot\cdots\odot v_n$ is by definition the sum of the degrees of $v_i$. For a permutation $\sigma\in\mathbb S_n$ and $v_1,\cdots, v_n\in V$,  the {\em Koszul sign} $\varepsilon(\sigma)=\varepsilon(\sigma;v_1,\cdots,v_n)\in\{-1,1\}$ is defined by
\begin{eqnarray*}
	v_1\odot\cdots\odot v_n=\varepsilon(\sigma;v_1,\cdots,v_n)v_{\sigma(1)}\odot\cdots\odot v_{\sigma(n)}.
\end{eqnarray*}
The {\em desuspension operator} $s^{-1}$ changes the grading of $V$ according to the rule $(s^{-1}V)^i:=V^{i+1}$. The  degree $-1$ map $s^{-1}:V\lon s^{-1}V$ is defined by sending $v\in V$ to its   copy $s^{-1}v\in s^{-1}V$.
\begin{defi}
An {\em  $L_\infty$-algebra} is a $\mathbb Z$-graded vector space $\g=\oplus_{k\in\mathbb Z}\g^k$ equipped with a collection $(k\ge 1)$ of linear maps $l_k:\otimes^k\g\lon\g$ of degree $1$ with the property that, for any homogeneous elements $x_1,\cdots,x_n\in \g$, we have
\begin{itemize}\item[\rm(i)]
{\em (graded symmetry)} for every $\sigma\in\mathbb S_{n}$,
\begin{eqnarray*}
l_n(x_{\sigma(1)},\cdots,x_{\sigma(n-1)},x_{\sigma(n)})=\varepsilon(\sigma)l_n(x_1,\cdots,x_{n-1},x_n),
\end{eqnarray*}
\item[\rm(ii)] {\em (generalized Jacobi identity)} for all $n\ge 1$,
\begin{eqnarray*}\label{sh-Lie}
\sum_{i=1}^{n}\sum_{\sigma\in \mathbb S_{(i,n-i)} }\varepsilon(\sigma)l_{n-i+1}(l_i(x_{\sigma(1)},\cdots,x_{\sigma(i)}),x_{\sigma(i+1)},\cdots,x_{\sigma(n)})=0.
\end{eqnarray*}
\end{itemize}
\end{defi}

\begin{defi}
An element $\alpha\in \g^0$ is called a {\em Maurer-Cartan element}  of an $L_\infty$-algebra $\g$ if $\alpha$ satisfies the Maurer-Cartan equation
\begin{eqnarray}\label{MC-equation}
\sum_{k=1}^{+\infty}\frac{1}{k!}l_k(\alpha,\cdots,\alpha)=0.
\end{eqnarray}
\end{defi}

Let $\alpha$ be a Maurer-Cartan element. Define $l_k^{\alpha}:\otimes^k\g\lon\g$  $(k\ge 1)$ by
\begin{eqnarray}
l_k^{\alpha}(x_1,\cdots,x_k)=\sum_{n=0}^{+\infty}\frac{1}{n!}l_{k+n}(\underbrace{\alpha,\cdots,\alpha}_n,x_1,\cdots,x_k).
\end{eqnarray}

\begin{rmk}To ensure the convergence of the series appearing in the definition of Maurer-Cartan elements and Maurer-Cartan twistings above, one need the $L_\infty$-algebra being {filtered} given by Dolgushev and Rogers in \cite{Dolgushev-Rogers}, or weakly filtered given in \cite{LST}. Since all the 	$L_\infty$-algebras under consideration in the sequel satisfy the weakly filtered condition, so we will not mention this point anymore.
\end{rmk}

The following result is given by Getzler in \cite[Section 4]{Getzler}.

\begin{thm}\label{deformation-mc}
With the above notation, $(\g,\{l_k^{\alpha}\}_{k=1}^{+\infty})$ is an $L_\infty$-algebra, obtained from $\g$ by twisting with the Maurer-Cartan element $\alpha$. Moreover, $\alpha+\alpha'$ is a Maurer-Cartan element of   $(\g,\{l_k\}_{k=1}^{+\infty})$ if and only if $\alpha'$ is a Maurer-Cartan element of the twisted $L_\infty$-algebra $(\g,\{l_k^{\alpha}\}_{k=1}^{+\infty})$.
\end{thm}

One method for constructing explicit $L_\infty$-algebras is given by Voronov's higher derived brackets \cite{Vo}. Let us recall this construction.

\begin{defi}
A $V$-data consists of a quadruple $(L,H,P,\Delta)$ where
\begin{itemize}
\item[$\bullet$] $(L,[\cdot,\cdot])$ is a graded Lie algebra,
\item[$\bullet$] $H$ is an abelian graded Lie subalgebra of $(L,[\cdot,\cdot])$,
\item[$\bullet$] $P:L\lon L$ is a projection, that is $P\circ P=P$, whose image is $H$ and kernel is a  graded Lie subalgebra of $(L,[\cdot,\cdot])$,
\item[$\bullet$] $\Delta$ is an element in $  \ker(P)^1$ such that $[\Delta,\Delta]=0$.
\end{itemize}
\end{defi}


\begin{thm}{\rm (\cite{Vo,Fregier-Zambon-1})}\label{thm:db-big-homotopy-lie-algebra}
Let $(L,H,P,\Delta)$ be a $V$-data. Then the graded vector space $s^{-1}L\oplus H$  is an $L_\infty$-algebra, where nontrivial $l_k$ are given by
\begin{eqnarray}
\nonumber l_1(s^{-1}x,a)&=&(-s^{-1}[\Delta,x],P(x+[\Delta,a])),\\
\nonumber l_2(s^{-1}x,s^{-1}y)&=&(-1)^xs^{-1}[x,y],\\
\nonumber l_k(s^{-1}x,a_1,\cdots,a_{k-1})&=&P[\cdots[[x,a_1],a_2]\cdots,a_{k-1}],\quad k\geq 2,\\
\label{V-shla}l_k(a_1,\cdots,a_{k-1},a_k)&=&P[\cdots[[\Delta,a_1],a_2]\cdots,a_{k}],\quad k\geq 2.
\end{eqnarray}
Here $a,a_1,\cdots,a_k$ are homogeneous elements of $H$ and $x,y$ are homogeneous elements of $L$.

Moreover, if $L'$ is a graded Lie subalgebra of $L$ that satisfies $[\Delta,L']\subset L'$, then $s^{-1}L'\oplus H$ is an $L_\infty$-subalgebra of the above $L_\infty$-algebra $(s^{-1}L\oplus H,\{l_k\}_{k=1}^{+\infty})$.
\end{thm}

\section{Deformations of relative Rota-Baxter Lie algebras}\label{sec:def}

In this section, first we use Voronov's higher derived brackets to construct the $L_\infty$-algebra whose Maurer-Cartan elements are relative \RB Lie algebra structures. Then using the twisting method, one obtains the $L_\infty$-algebra that controls simultaneous deformations of relative \RB Lie algebras.

\begin{defi}
  A {\em \LR pair}, denoted by $(\g,\mu;\rho)$, consists of a Lie algebra  $(\g,\mu=[\cdot,\cdot]_\g)$  and a representation $\rho:\g\longrightarrow\gl(V)$   of $\g$ on a vector space $V$.
\end{defi}

Note that $\mu+\rho\in C^{1|0}(\g\oplus V,\g\oplus V)$. Moreover, the fact that $\mu$ is a Lie bracket and $\rho$ is a representation is equivalent to that
$$
[\mu+\rho,\mu+\rho]_\NR=0.
$$

We now recall the notion of a relative \RB operator.

\begin{defi} \label{defi:O}
\begin{enumerate}
\item[\rm(i)] A linear operator $T:\g\longrightarrow \g$ on a Lie algebra $\g$ is called a {\em \RB operator} if
\begin{equation} [T(x),T(y)]_\g=T\big([T(x),y]_\g+ [x,T(y)]_\g \big), \quad \forall x, y \in \g.
\label{eq:rbo}
\end{equation}
Moreover, a Lie algebra $(\g,[\cdot,\cdot]_\g)$ with a \RB operator $T$ is
called a {\em \RB Lie algebra}, which is denoted   by $(\g,[\cdot,\cdot]_\g,T)$.
\item[\rm(ii)] A {\em relative \RB Lie algebra} is a triple $((\g,[\cdot,\cdot]_\g),\rho,T)$, where $(\g,[\cdot,\cdot]_\g;\rho)$ is a \LR pair and $T:V\longrightarrow\g$ is a  {\em relative \RB operator}, i.e.
 \begin{equation}
   [Tu,Tv]_\g=T\big(\rho(Tu)(v)-\rho(Tv)(u)\big),\quad\forall u,v\in V.
 \mlabel{eq:defiO}
 \end{equation}

\end{enumerate}
\end{defi}

Note that a \RB operator on a Lie algebra is a relative \RB operator with respect to the adjoint representation.

Let $\g$ and $V$ be two vector spaces. Then one has a graded Lie algebra $(\oplus_{n=0}^{+\infty}C^{n}(\g\oplus V,\g\oplus V),[\cdot,\cdot]_{\NR})$. This graded Lie algebra gives rise to a V-data, and an $L_\infty$-algebra naturally.
\begin{pro}\label{pro:VdataL}
One has a $V$-data $(L,H,P,\Delta)$ as follows:
\begin{itemize}
\item[$\bullet$] the graded Lie algebra $(L,[\cdot,\cdot])$ is given by $\big(\oplus_{n=0}^{+\infty}C^{n}(\g\oplus V,\g\oplus V),[\cdot,\cdot]_{\NR}\big);$
\item[$\bullet$] the abelian graded Lie subalgebra $H$ is given by
\begin{equation}\label{defi:h}
H:=\oplus_{n=0}^{+\infty}C^{-1|(n+1)}(\g\oplus V,\g\oplus V)=\oplus_{n=0}^{+\infty}\Hom(\wedge^{n+1}V,\g);
\end{equation}
\item[$\bullet$] $P:L\lon L$ is the projection onto the subspace $H;$
\item[$\bullet$] $\Delta=0$.
\end{itemize}

Consequently, one obtains an $L_\infty$-algebra $(s^{-1}L\oplus H,\{l_k\}_{k=1}^{+\infty})$, where $l_k$ are given by
\begin{eqnarray*}
l_1(s^{-1}Q,\theta)&=&P(Q),\\
  l_2(s^{-1}Q,s^{-1}Q')&=&(-1)^Qs^{-1}[Q,Q']_{\NR}, \\
l_k(s^{-1}Q,\theta_1,\cdots,\theta_{k-1})&=&P[\cdots[Q,\theta_1]_{\NR},\cdots,\theta_{k-1}]_{\NR},
\end{eqnarray*}
for homogeneous elements   $\theta,\theta_1,\cdots,\theta_{k-1}\in H$, homogeneous elements  $Q,Q'\in L$ and all the other possible combinations vanish.
\end{pro}
\begin{proof}
 By Lemma \ref{Zero-condition-2},    $H$ is an abelian subalgebra of $(L,[\cdot,\cdot])$.

   Since $P$ is the projection onto $H$, it is obvious that $P\circ P=P$. It is also straightforward to see that the kernel  of $P$ is a  graded Lie subalgebra of $(L,[\cdot,\cdot])$. Thus $(L,H,P,\Delta=0)$ is a V-data.

   The other conclusions follows immediately from Theorem \ref{thm:db-big-homotopy-lie-algebra}.
\end{proof}
By Lemma \ref{important-lemma-2}, one obtains that
\begin{equation}\label{defi:Lprime}
L'=\oplus_{n=0}^{+\infty}C^{n|0}(\g\oplus V,\g\oplus V),\quad\mbox{where}\quad C^{n|0}(\g\oplus V,\g\oplus V)=\Hom(\wedge^{n+1}\g,\g)\oplus\Hom(\wedge^{n}\g\otimes V,V)
 \end{equation}is a graded Lie subalgebra of $\big(\oplus_{n=0}^{+\infty}C^{n}(\g\oplus V,\g\oplus V),[\cdot,\cdot]_{\NR}\big)$.

\begin{cor}\label{cor:Linfty}
 With the above notation,  $(s^{-1}L'\oplus H,\{l_i\}_{i=1}^{+\infty})$ is an $L_\infty $-algebra, where $l_k$ are given by
 \begin{eqnarray*}
  l_2(s^{-1}Q,s^{-1}Q')&=&(-1)^Qs^{-1}[Q,Q']_{\NR}, \\
l_k(s^{-1}Q,\theta_1,\cdots,\theta_{k-1})&=&P[\cdots[Q,\theta_1]_{\NR},\cdots,\theta_{k-1}]_{\NR},
\end{eqnarray*}
for homogeneous elements   $\theta_1,\cdots,\theta_{k-1}\in H$, homogeneous elements  $Q,Q'\in L'$, and all the other possible combinations vanish.

\end{cor}

Now we are ready to formulate the main result in this section.

\begin{thm}\label{deformation-rota-baxter}
  Let $\g$ and $V$ be two vector spaces,  $\mu\in\Hom(\wedge^2\g,\g),~\rho\in\Hom(\g\otimes V,V)$ and $T\in\Hom(V,\g)$. Then $((\g,\mu),\rho,T)$ is a relative \RB Lie algebra if and only if  $(s^{-1}\pi,T)$ is a Maurer-Cartan element of the $L_\infty$-algebra $(s^{-1}L'\oplus H,\{l_i\}_{i=1}^{+\infty})$ given in Corollary \ref{cor:Linfty}, where $\pi=\mu+\rho\in C^{1|0}(\g\oplus V,\g\oplus V)$.
\end{thm}

\begin{proof}
  Let $(s^{-1}\pi,T)$ be a Maurer-Cartan element of $(s^{-1}L'\oplus H,\{l_i\}_{i=1}^{+\infty})$. By Lemma \ref{important-lemma-2} and Lemma \ref{Zero-condition-2}, we have
\begin{eqnarray*}
||[\pi,T]_{\NR}||=0|1,\quad ||[[\pi,T]_{\NR},T]_{\NR}||=-1|2,\quad [[[\pi,T]_{\NR},T]_{\NR},T]_{\NR}=0.
\end{eqnarray*}
Then, by Corollary \ref{cor:Linfty}, we have
\begin{eqnarray*}
(0,0)&=&\sum_{k=1}^{+\infty}\frac{1}{k!}l_k\Big((s^{-1}\pi,T),\cdots,(s^{-1}\pi,T)\Big)\\
&=&\frac{1}{2!}l_2\Big((s^{-1}\pi,T),(s^{-1}\pi,T)\Big)+\frac{1}{3!}l_3\Big((s^{-1}\pi,T),(s^{-1}\pi,T),(s^{-1}\pi,T)\Big)\\
&=&\Big(-s^{-1}\frac{1}{2}[\pi,\pi]_{\NR},\frac{1}{2}[[\pi,T]_{\NR},T]_{\NR}\Big).
\end{eqnarray*}
Thus, we obtain $
 ~[\pi,\pi]_{\NR}=0$ and $
[[\pi,T]_{\NR},T]_{\NR}=0,
$
 which implies that  $(\g,\mu)$ is a Lie algebra, $(V;\rho)$ is its representation and   $T$ is a relative \RB operator on  the Lie algebra $(\g,\mu)$ with respect to the representation $(V;\rho)$.
\end{proof}

Let $((\g,\mu),\rho,T)$ be a relative \RB Lie algebra. Denote by $\pi=\mu+\rho\in C^{1|0}(\g\oplus V,\g\oplus V)$. By Theorem \ref{deformation-rota-baxter}, we obtain that $(s^{-1}\pi,T)$ is a Maurer-Cartan element of the $L_\infty$-algebra  $(s^{-1}L'\oplus H,\{l_i\}_{i=1}^{+\infty})$ given in Corollary \ref{cor:Linfty}.
 Now  we are ready to give the
  $L_\infty$-algebra that controls deformations of the relative \RB Lie algebra.

\begin{thm}\label{thm:Simultaneous-deformation}
With the above notation, one has the twisted $L_\infty$-algebra $\big(s^{-1}L'\oplus H,\{l_k^{(s^{-1}\pi,T)}\}_{k=1}^{+\infty}\big)$ associated to a relative \RB Lie algebra $((\g,\mu),\rho,T)$, where $\pi=\mu+\rho$.

Moreover, for
linear maps $T'\in\Hom(V,\g)$, $\mu'\in\Hom(\wedge^2\g,\g)$ and $\rho'\in\Hom(\g,\gl(V))$, the triple $((\g,\mu+\mu'),\rho+\rho',T+T')$ is again a relative \RB Lie algebra if and only if $\big(s^{-1}(\mu'+\rho'),T'\big)$ is a Maurer-Cartan element of the twisted $L_\infty$-algebra $\big(s^{-1}L'\oplus H,\{l_k^{(s^{-1}\pi,T)}\}_{k=1}^{+\infty}\big)$.
\end{thm}

\begin{proof}
If $((\g,\mu+\mu'),\rho+\rho',T+T')$ is   a relative \RB Lie algebra, then
by Theorem \ref{deformation-rota-baxter},  $(s^{-1}(\mu+\mu'+\rho+\rho'),T+T')$ is a Maurer-Cartan
element of the $L_\infty$-algebra given in Corollary \ref{cor:Linfty}. Moreover, by Theorem \ref{deformation-mc}, $(s^{-1}(\mu'+\rho'),T')$ is a Maurer-Cartan
element of the $L_\infty$-algebra $\big(s^{-1}L'\oplus H,\{l_k^{(s^{-1}\pi,T)}\}_{k=1}^{+\infty}\big)$.
\end{proof}

\begin{rmk}
  The above $L_\infty$-algebra controlling deformations of relative \RB Lie algebras is an extension of the differential graded Lie algebra controlling deformations of \LR pairs by the differential graded Lie algebra controlling deformations of relative \RB operators. See \cite[Theorem 3.16]{LST} for more details.
\end{rmk}

\section{Cohomologies of relative Rota-Baxter Lie algebras}\label{sec:coh}

In this section, we survey the cohomology of relative \RB Lie algebras. In particular, one can define  the cohomology of \RB Lie algebras and the cohomology of triangular Lie bialgebras using this general framework.

One can define the cohomology of a relative \RB Lie algebra using the twisted $L_\infty$-algebra given in Theorem \ref{thm:Simultaneous-deformation}.

Let   $((\g,\mu),\rho,T)$ be a relative \RB Lie algebra. Define the set of  $0$-cochains $\frkC^0(\g,\rho,T)$ to be $0$, and define
the set of  $1$-cochains $\frkC^1(\g,\rho,T)$ to be $\gl(\g)\oplus \gl(V)$. For $n\geq 2$, define the space of    $n$-cochains $\frkC^n(\g,\rho,T)$ by
\begin{eqnarray*}
\frkC^n(\g,\rho,T)&:=&\frkC^n(\g,\rho)\oplus \frkC^n(T)=C^{(n-1)|0}(\g\oplus V,\g\oplus V)\oplus C^{-1|(n-1)}(\g\oplus V,\g\oplus V)\\
&=&\Big(\Hom(\wedge^n\g,\g)\oplus \Hom(\wedge^{n-1}\g\otimes V,V)\Big)\oplus\Hom(\wedge^{n-1}V,\g).
\end{eqnarray*}

Define the {\em coboundary operator} $\huaD:\frkC^n(\g,\rho,T)\lon \frkC^{n+1}(\g,\rho,T)$ by
\begin{equation}\label{cohomology-of-RB}
  \huaD(f,\theta)=(-1)^{n-2}\big(-[\pi,f]_{\NR},[[\pi,T]_{\NR},\theta]_{\NR}+\frac{1}{n!}\underbrace{[\cdots[[}_nf,T]_{\NR},T]_{\NR},\cdots,T]_{\NR}\big),
\end{equation}
where $\pi=\mu+\rho, ~f\in\Hom(\wedge^n\g,\g)\oplus \Hom(\wedge^{n-1}\g\otimes V,V)$ and $\theta\in \Hom(\wedge^{n-1}V,\g).$

\begin{thm}\label{cohomology-of-relative-RB}
  With the above notation,  $(\oplus _{n=0}^{+\infty}\frkC^n(\g,\rho,T),\huaD)$ is a cochain complex, i.e. $\huaD\circ \huaD=0.$
\end{thm}
\begin{proof}
By Theorem \ref{thm:Simultaneous-deformation}, $\big(s^{-1}L'\oplus H,\{l_k^{(s^{-1}\pi,T)}\}_{k=1}^{+\infty}\big)$ is an $L_\infty$-algebra, where $\pi=\mu+\rho$, $H$ and $L'$ are given by \eqref{defi:h} and \eqref{defi:Lprime} respectively. For any $(f,\theta)\in\frkC^n(\g,\rho,T)$, one has $(s^{-1}f,\theta)\in (s^{-1}L'\oplus H)^{n-2}$. By \eqref{cohomology-of-RB}, one deduces that
$$
\huaD(f,\theta)=(-1)^{n-2} l_1^{(s^{-1}\pi,T)}(s^{-1}f,\theta).
$$
Thus, $(\oplus _{n=0}^{+\infty}\frkC^n(\g,\rho,T),\huaD)$ is a cochain complex.
\end{proof}

\begin{defi}
  The cohomology of the cochain complex $(\oplus _{n=0}^{+\infty}\frkC^n(\g,\rho,T),\huaD)$ is called the {\em cohomology  of the relative \RB Lie algebra} $((\g,\mu),\rho,T)$. We denote its $n$-th cohomology group by $\huaH^n(\g,\rho,T)$.
\end{defi}

Define a linear operator $h_T:\frkC^n(\g,\rho)\lon \frkC^{n+1}(T)$ by
\begin{eqnarray}\label{key-cohomology-T-abstract}
h_Tf:=(-1)^{n-2} \frac{1}{n!}\underbrace{[\cdots[[}_nf,T]_{\NR},T]_{\NR},\cdots,T]_{\NR}.
\end{eqnarray}
More precisely, \begin{eqnarray}\label{key-cohomology-T}
&&\nonumber( h_Tf)(v_1,\cdots,v_n)\\
&=&(-1)^{n}f_\g(Tv_1,\cdots,Tv_n)+\sum_{i=1}^{n}(-1)^{i+1}Tf_V\big(Tv_1,\cdots,Tv_{i-1},Tv_{i+1},\cdots,Tv_n,v_i\big),
\end{eqnarray}
where $f=(f_\g,f_V),$ and $f_\g\in \Hom(\wedge^n\g,\g),~f_V\in \Hom(\wedge^{n-1}\g\otimes V,V)$ and $v_1,\cdots,v_n\in V.$

By \eqref{cohomology-of-RB} and \eqref{key-cohomology-T-abstract}, the coboundary operator can be written as
\begin{equation}\label{eq:Dexplicit}
\huaD(f,\theta)=(\partial f,\delta \theta+h_Tf),
\end{equation}
where $\partial:\frkC^n(\g,\rho)\lon \frkC^{n+1}(\g,\rho)$ is given by
\begin{equation}\label{defi:coboundary-rep}
  \partial f:=(-1)^{n-1}[ {\mu}+ {\rho}, {f}]_{\NR}.
\end{equation}
and $\delta:\frkC^n(T)\to \frkC^{n+1}(T) $ is given by
\begin{eqnarray}\label{eq:odiff}
   \delta \theta=(-1)^{n}[[\pi,T]_{\NR},\theta]_{\NR}.
\end{eqnarray}

The formula of the coboundary operator $\huaD$ can be well-explained by the following diagram:
 \[
\small{ \xymatrix{
\cdots
\longrightarrow \frkC^n(\g,\rho)\ar[dr]^{h_T} \ar[r]^{\qquad\partial} & \frkC^{n+1}(\g,\rho) \ar[dr]^{h_T} \ar[r]^{\partial\qquad}  & \frkC^{n+2}(\g,\rho)\longrightarrow\cdots  \\
\cdots\longrightarrow \frkC^n(T) \ar[r]^{\qquad\delta} &\frkC^{n+1}(T)\ar[r]^{\delta\qquad}&\frkC^{n+2}(T)\longrightarrow \cdots.}
}
\]

Since $\huaD^2=0$, it follows that $\partial^2=0$ and $\delta^2=0$. Therefore, one has two cochain complexes $(\oplus_{n=0}^{+\infty}\frkC^n(\g,\rho),\partial)$ and $(\oplus_{n=0}^{+\infty}\frkC^n(T),\delta)$, whose cohomology are denoted by $ \huaH^*(\g,\rho)$ and  $\huaH^*(T)$ respectively.

\begin{thm}\label{cohomology-exact}
Let $((\g,\mu),\rho,T)$ be a relative \RB Lie algebra. Then there is a short exact sequence of the  cochain complexes:
$$
0\longrightarrow(\oplus_{n=0}^{+\infty}\frkC^n(T),\delta)\stackrel{\iota}{\longrightarrow}(\oplus _{n=0}^{+\infty}\frkC^n(\g,\rho,T),\huaD)\stackrel{p}{\longrightarrow} (\oplus_{n=0}^{+\infty}\frkC^n(\g,\rho),\partial)\longrightarrow 0,
$$
where $\iota$ and $p$ are the inclusion map and the projection map.

Consequently, there is a long exact sequence of the  cohomology groups:
$$
\cdots\longrightarrow\huaH^n(T)\stackrel{\huaH^n(\iota)}{\longrightarrow}\huaH^n(\g,\rho,T)\stackrel{\huaH^n(p)}{\longrightarrow} \huaH^n(\g,\rho)\stackrel{c^n}\longrightarrow \huaH^{n+1}(T)\longrightarrow\cdots,
$$
where the connecting map $c^n$ is defined by
$
c^n([\alpha])=[h_T\alpha],$  for all $[\alpha]\in \huaH^n(\g,\rho).$
\end{thm}
\begin{proof}
 By  \eqref{eq:Dexplicit}, one has the short exact sequence  of cochain complexes which induces a long exact sequence of cohomology groups.   Also by \eqref{eq:Dexplicit},   $c^n$ is given by
$
c^n([\alpha])=[h_T\alpha].$
\end{proof}

\begin{rmk}
  The cohomology of the cochain complex $(\oplus_{n=0}^{+\infty}\frkC^n(T),\delta)$ is taken to be the cohomology of the relative \RB operator $T$ \cite{TBGS-1}, and the cohomology of the cochain complex $(\oplus_{n=0}^{+\infty}\frkC^n(\g,\rho),\partial)$ is taken to be the cohomology of the  \LR  pair $(\g,\mu;\rho)$  \cite{Arnal}. So the above result establishes the relationship between the cohomology groups of relative \RB Lie algebras and the cohomology groups of the underlying relative \RB operators and \LR pairs.
\end{rmk}

\begin{rmk}
   In the associative algebra context, deformations, cohomologies and homotopies of relative \RB operators on associative algebras and relative \RB associative algebras were studied in \cite{Das} and \cite{DasM} respectively.
\end{rmk}

\subsection{Cohomology of Rota-Baxter Lie algebras}\label{sec:cohomologyRB}

In this subsection,  we survey the cohomology of \RB Lie algebras, which is defined with the help of the general framework of the cohomology of relative \RB Lie algebras.

Let $(\g,[\cdot,\cdot]_\g,T)$ be a \RB Lie algebra. Define the set of $0$-cochains $\frkC^0_{{\rm RB}}(\g,T)$ to be $0$, and define
the set of $1$-cochains $\frkC^1_{{\rm RB}}(\g,T)$ to be $\frkC^1_{{\rm RB}}(\g,T):=\Hom( \g,\g)$. For $n\geq2$, define the space of   $n$-cochains $\frkC^n_{{\rm RB}}(\g,T)$ by
\begin{eqnarray*}
\frkC^n_{\rm RB}(\g,T):=\frkC^n_\Lie(\g;\g)\oplus \frkC^n(T) =\Hom(\wedge^n\g,\g)\oplus\Hom(\wedge^{n-1}\g,\g).
\end{eqnarray*}

Define the embedding $\frki:\frkC^n_{\rm RB}(\g,T)\lon \frkC^n(\g,\ad,T)$ by
$$
\frki(f,\theta)=(f,f,\theta),\quad \forall f\in \Hom(\wedge^n\g,\g), \theta\in\Hom(\wedge^{n-1}\g,\g).
$$
Denote by $\Img^n(\frki)=\frki(\frkC^n_{\rm RB}(\g,T))$. Then   $(\oplus_{n=0}^{+\infty}\Img^n(\frki),\huaD)$ is a subcomplex of the cochain complex $(\oplus _{n=0}^{+\infty}\frkC^n(\g,\ad,T),\huaD)$ associated to the relative \RB Lie algebra $((\g,[\cdot,\cdot]_\g),\ad,T)$.

Define the projection $\frkp:\Img^n(\frki)\lon \frkC^n_{\rm RB}(\g,T)$ by
$$
\frkp(f,f,\theta)=(f,\theta), \quad \forall f  \in \Hom(\wedge^n\g,\g), \theta\in\Hom(\wedge^{n-1}\g,\g).
$$
Then for $n\geq0,$  define $\huaD_{\rm RB}:\frkC^n_{\rm RB}(\g,T)\lon \frkC^{n+1}_{\rm RB}(\g,T)$ by
$
\huaD_{\rm RB}=\frkp\circ \huaD\circ \frki.
$
More precisely,
\begin{eqnarray}\label{eq:dRB}
\huaD_{\rm RB}(f,\theta)=\Big({\dM_\CE} f, \delta \theta +\Omega f\Big),\quad \forall f\in \Hom(\wedge^n\g,\g),~\theta\in \Hom(\wedge^{n-1}\g,\g),
\end{eqnarray}
where  $\delta$ is given by \eqref{eq:odiff} and $\Omega:\Hom(\wedge^n\g,\g)\lon \Hom(\wedge^n\g,\g)$ is defined  by
\begin{eqnarray*}
(\Omega f)(x_1,\cdots,x_n)=(-1)^{n}\Big(f(Tx_1,\cdots,Tx_n)-\sum_{i=1}^{n}Tf(Tx_1,\cdots,Tx_{i-1},x_i,Tx_{i+1},\cdots,Tx_n)\Big).
\end{eqnarray*}
\begin{thm}
The map $\huaD_{\rm RB}$ is a coboundary operator, i.e. $\huaD_{\rm RB}\circ\huaD_{\rm RB}=0$.
\end{thm}
\begin{proof}
One has
  \begin{eqnarray*}
   \huaD_{\rm RB}\circ\huaD_{\rm RB}=\frkp\circ \huaD\circ \frki\circ \frkp\circ \huaD\circ \frki=\frkp\circ \huaD\circ   \huaD\circ \frki=0,
  \end{eqnarray*}
 which finishes the proof.
\end{proof}

\begin{defi}
  Let $(\g,[\cdot,\cdot]_\g,T)$ be a \RB Lie algebra. The cohomology of the cochain complex  $(\oplus_{n=0}^{+\infty}\frkC^n_{\rm RB}(\g,T),\huaD_{\rm RB})$ is taken to be the {\em cohomology of the \RB Lie algebra} $(\g,[\cdot,\cdot]_\g,T)$. Denote the $n$-th cohomology group by $\huaH^n_{\rm RB}(\g,T).$
\end{defi}

\subsection{Cohomology of triangular Lie bialgebras}\label{sec:cohomologyTLB}
In this subsection, all vector spaces are assumed to be finite-dimensional. We survey the cohomology of   triangular  Lie bialgebras, which is defined  with the help of the general cohomological framework for relative \RB Lie algebras.

Recall that a Lie bialgebra is a vector space $\g$ equipped
with a Lie algebra structure
$[\cdot,\cdot]_\g:\wedge^2\g\longrightarrow\g$ and a Lie coalgebra
structure $\delta:\g\longrightarrow\wedge^2\g$ such that $\delta$
is a 1-cocycle on $\g$ with coefficients in $\wedge^2\g$.
The Lie bracket $[\cdot,\cdot]_\g$ in a Lie algebra $\g$
naturally extends to the  {\em Schouten-Nijenhuis bracket} $[\cdot,\cdot]_\SN$ on $\wedge^\bullet\g=\oplus_{k\geq 0} \wedge^{k+1}\g$. More precisely, one has
$$
~[x_1\wedge\cdots \wedge x_p,y_1\wedge\cdots\wedge y_q]_\SN
=\sum_{1\le i\le p\atop 1\le j\le q} (-1)^{i+j}[x_i,y_j]_\g\wedge x_1\wedge\cdots\hat{x}_i\cdots \wedge x_p\wedge y_1\wedge\cdots\hat{y}_j\cdots\wedge y_q.
$$

An element $r\in\wedge^2\g$ is called a {\em skew-symmetric $r$-matrix} \cite{STS} if it satisfies the {\em classical Yang-Baxter equation}
$
[r,r]_\SN=0
$.
It is well known ~\cite{Ku} that $r$ satisfies the classical Yang-Baxter
equation if and only if $r^\sharp$ is a relative \RB operator on $\g$
with respect to the coadjoint representation,
where $r^\sharp:\g^*\lon\g$ is defined by $\langle r^\sharp(\xi),\eta\rangle=\langle r,\xi\wedge \eta\rangle$ for all $\xi,\eta\in\g^*$.

Let $r$ be a skew-symmetric $r$-matrix. Define
$\delta_r:\g\longrightarrow\wedge^2\g$ by
$
\delta_r(x)=[x,r]_\SN,$  for all $x\in\g.
$
Then $(\g,[\cdot,\cdot]_\g,\delta_r)$ is a Lie bialgebra, which is
called a {\em triangular Lie bialgebra}. From now on,   denote a triangular Lie bialgebra by $(\g,[\cdot,\cdot]_\g,r)$.

Let $\g$ be a Lie algebra and $r\in\wedge^2\g$ a skew-symmetric $r$-matrix. Define the set of 0-cochains and 1-cochains to be zero and define the set of $k$-cochains to be $\wedge^k\g$. Define $\dM_r:\wedge^k\g\lon\wedge^{k+1}\g$ by
\begin{equation}\label{eq:dr}
  \dM_r \chi =[r,\chi]_\SN,\quad\forall \chi\in \wedge^k\g.
\end{equation}
Then $\dM_r^2=0.$
Denote by $\huaH^k(r)$ the corresponding $k$-th cohomology group,
called  the {\em $k$-th cohomology group of the skew-symmetric
$r$-matrix $r$}.

For any $k\geq 1$, define $\Psi:\wedge^{k+1}\g\longrightarrow \Hom(\wedge^k\g^*,\g)$ by
\begin{equation}\label{eq:defipsi}
 \langle\Psi(\chi)(\xi_1,\cdots,\xi_k),\xi_{k+1}\rangle=\langle \chi,\xi_1\wedge\cdots\wedge\xi_k\wedge\xi_{k+1}\rangle,\quad \forall \chi\in\wedge^{k+1}\g, \xi_1,\cdots, \xi_{k+1}\in\g^*.
\end{equation}
By \cite[Theorem 7.7]{TBGS-1}, we have
\begin{equation}\label{eq:relationdd}
  \Psi(\dM_r\chi)=\delta(\Psi(\chi)),\quad \forall \chi\in\wedge^k\g.
\end{equation}
Thus $(\Img(\Psi),\delta)$ is a subcomplex of the cochain complex $(\oplus_k\frkC^k(r^\sharp),\delta)$ associated to the relative \RB operator $r^\sharp$, where $\Img(\Psi):=\oplus_k\{\Psi(\chi)|\forall \chi\in\wedge^k\g\}$ and $\delta$ is the coboundary operator given by \eqref{eq:odiff} for the relative \RB operator $r^\sharp$.

In the following, we survey the cohomology of a triangular Lie bialgebra $(\g,[\cdot,\cdot]_\g,r)$. Define the set of  $0$-cochains $\frkC^0_{\TLB}(\g,r)$ to be $0$, and define
the set of  $1$-cochains to be $\frkC^1_{\TLB}(\g,r):=\Hom(\g,\g)$. For $n\geq 2$, define the space of   $n$-cochains $\frkC^n_{\TLB}(\g,r)$ by
\begin{eqnarray*}
\frkC^n_{\TLB}(\g,r):=\Hom(\wedge^n\g,\g)\oplus \wedge^{n}\g.
\end{eqnarray*}

Define the embedding $\frki:\frkC^n_{\TLB}(\g,r)\lon \frkC^n(\g,\ad^*,r^{\sharp})=\Hom(\wedge^n\g,\g)\oplus \Hom(\wedge^{n-1}\g\otimes \g^*,\g^*)\oplus \Hom(\wedge^{n-1}\g^*,\g)$ by
$$
\frki(f,\chi)=(f,f^\star,\Psi(\chi)),\quad \forall f\in \Hom(\wedge^n\g,\g), \chi\in\wedge^n\g,
$$
where  $f^\star\in\Hom(\wedge^{n-1}\g\otimes\g^*,\g^*)$ is defined by
\begin{eqnarray}
 \label{dual-lie-r-1}\langle f^\star(x_1,\cdots,x_{n-1},\xi),x_n\rangle=-\langle \xi,f(x_1,\cdots,x_{n-1},x_n)\rangle.
\end{eqnarray}

Denote by $\Img^n(\frki)$ the image of $\frki$, i.e. $\Img^n(\frki):=\{\frki(f,\chi)|\forall (f,\chi)\in \frkC^n_{\TLB}(\g,r)\}$.
It was proved in \cite{LST} that  $(\oplus_n\Img^n(\frki),\huaD)$ is a subcomplex of the cochain complex $(\frkC^n(\g,\ad^*,r^{\sharp}),\huaD)$ associated to the relative \RB Lie algebra $((\g,[\cdot,\cdot]_\g),\ad^*,r^\sharp)$.

Define the projection $\frkp:\Img^n(\frki)\lon \frkC^n_{\TLB}(\g,r)$ by
$$
\frkp(f,f^\star,\theta)=(f,\theta^\flat), \quad \forall f \in \Hom(\wedge^n\g,\g),~ \theta\in\{\Psi(\chi)|\forall \chi\in\wedge^n\g\},
$$
where $\theta^\flat\in\wedge^n\g$ is defined by$
\langle\theta^\flat,\xi_1\wedge\cdots\wedge\xi_n\rangle=\langle\theta(\xi_1,\cdots,\xi_{n-1}),\xi_n\rangle.
$
Define the {\em coboundary operator} $\huaD_\TLB:\frkC^n_{\TLB}(\g,r)\lon \frkC^{n+1}_{\TLB}(\g,r)$ for a triangular Lie bialgebra  by
$$
\huaD_\TLB=\frkp\circ \huaD\circ \frki.
$$

\begin{thm}
The map $\huaD_\TLB$ is a coboundary operator, i.e. $\huaD_\TLB\circ\huaD_\TLB=0$.
\end{thm}
\begin{proof}
 Since $\frki\circ\frkp=\Id$ when restricting on the image of $\frki$, one has
  \begin{eqnarray*}
   \huaD_\TLB\circ\huaD_\TLB=\frkp\circ \huaD\circ \frki\circ \frkp\circ \huaD\circ \frki=\frkp\circ \huaD\circ   \huaD\circ \frki=0,
  \end{eqnarray*}
 which finishes the proof.
\end{proof}

 \begin{defi}
  Let $(\g,[\cdot,\cdot]_\g,r)$ be a triangular Lie bialgebra. The cohomology of the cochain complex  $(\oplus_{n=0}^{+\infty}\frkC^n_{\TLB}(\g,r),\huaD_\TLB)$ is called  the {\em cohomology of the triangular Lie bialgebra} $(\g,[\cdot,\cdot]_\g,r)$. Denote the $n$-th cohomology group by $ \huaH^n_\TLB(\g,r)$.
\end{defi}

Now we give the  precise formula for the coboundary operator $\huaD_\TLB$. By the definition of $\frki$, $\frkp$, $\huaD$ and \eqref{eq:relationdd}, one has
\begin{eqnarray}\label{eq:dTLBexplicit}
\huaD_\TLB(f,\chi)=\Big({\dM}_\CE f, \Theta f+\dM_r \chi \Big),\qquad \forall f\in \Hom(\wedge^n\g,\g),~\chi\in \wedge^{n}\g,
\end{eqnarray}
where   $\dM_r$ is given by \eqref{eq:dr} and $\Theta:\Hom(\wedge^n\g,\g)\lon  \wedge^{n+1}\g $ is defined  by
$
\Theta f=\Psi^{-1}(h_{r^\sharp}(f,f^\star)).
$
More precisely,
 \begin{eqnarray}
   \langle \Theta f,\xi_1\wedge \cdots\wedge\xi_{n+1}\rangle=\sum_{i=1}^{n+1}(-1)^{i+1}\langle\xi_i,f(r^{\sharp}(\xi_1),\cdots,r^{\sharp}(\xi_{i-1}),r^{\sharp}(\xi_{i+1}),
   \cdots,r^{\sharp}(\xi_{n+1}))\rangle,
 \end{eqnarray}
for all $f\in\Hom(\wedge^n\g,\g)$ and $\xi_1,\cdots,\xi_{n+1}\in\g^*$.

\begin{rmk}
  One can use the cohomology theory developed here to study infinitesimal deformations. More precisely, the cohomology groups $\huaH^2(\g,\rho,T),~\huaH^2_{\rm RB}(\g,T), ~\huaH^2_\TLB(\g,r)$ classify infinitesimal deformations of the relative \RB Lie algebra $(\g,\rho,T)$, the \RB Lie algebra $(\g,T)$ and the triangular Lie bialgebra $(\g,r)$ respectively.
\end{rmk}

\section{Homotopies of relative Rota-Baxter Lie algebras}\label{sec:hom}
In this section, we survey  the notion of a homotopy relative \RB Lie algebra, which consists of an $L_\infty$-algebra, its representation and a homotopy relative \RB operator.  Homotopy relative \RB operators can be characterized as Maurer-Cartan elements in a certain $L_\infty$-algebra. 


Denote by $\Hom^n(\bar{\Sym}(V),V)$ the space of degree $n$ linear maps from the graded vector space $\bar{\Sym}(V)=\oplus_{i=1}^{+\infty}\Sym^{i}(V)$ to the $\mathbb Z$-graded vector space $V$. Obviously, an element $f\in\Hom^n(\bar{\Sym}(V),V)$ is the sum of $f_i:\Sym^i(V)\lon V$. We will write  $f=\sum_{i=1}^{+\infty} f_i$.
 Set $C^n(V,V):=\Hom^n(\bar{\Sym}(V),V)$ and
$
C^*(V,V):=\oplus_{n\in\mathbb Z}C^n(V,V).
$
As the graded version of the Nijenhuis-Richardson bracket given in \cite{NR,NR2}, the {\em graded Nijenhuis-Richardson bracket} $[\cdot,\cdot]_{\NR}$ on the graded vector space $C^*(V,V)$ is given
by
\begin{eqnarray}
[f,g]_{\NR}:=f\bar{\circ} g-(-1)^{mn}g\bar{\circ}f,\,\,\,\,\forall f=\sum_{i=1}^{+\infty} f_i\in C^m(V,V),~g=\sum_{j=1}^{+\infty} g_j\in C^n(V,V),
\label{eq:gfgcirc-lie}
\end{eqnarray}
where $f\bar{\circ}g\in C^{m+n}(V,V)$ is defined by
 \begin{eqnarray}\label{NR-circ}
f\bar{\circ}g&=&\Big(\sum_{i=1}^{+\infty}f_i\Big)\bar{\circ}\Big(\sum_{j=1}^{+\infty}g_j\Big):=\sum_{k=1}^{+\infty}\Big(\sum_{i+j=k+1}f_i\bar{\circ} g_j\Big),
\end{eqnarray}
while $f_i\bar{\circ} g_j\in \Hom(\Sym^{i+j-1}(V),V)$ is defined by
\begin{eqnarray}\label{graded-NR}
(f_i\bar{\circ} g_j)(v_1,\cdots,v_{i+j-1})
:=\sum_{\sigma\in\mathbb S_{(j,i-1)}}\varepsilon(\sigma)f_i(g_j(v_{\sigma(1)},\cdots,v_{\sigma(j)}),v_{\sigma(j+1)},\cdots,v_{\sigma(i+j-1)}).
\end{eqnarray}

The following result is well-known and, in fact, can be taken as a definition of an $L_\infty$-algebra.
\begin{thm}\label{graded-Nijenhuis-Richardson-bracket}
With the above notation, $(C^*(V,V),[\cdot,\cdot]_{\NR})$ is a graded Lie algebra. Its Maurer-Cartan elements $\sum_{k=1}^{+\infty}l_k$ are the $L_\infty$-algebra structures on $V$.
\end{thm}

\begin{defi}{\rm (\cite{LM})}
A {\em representation} of an $L_\infty$-algebra $(\g,\{l_k\}_{k=1}^{+\infty})$ on a graded vector space $V$ consists of linear maps $\rho_k:\Sym^{k-1}(\g)\otimes V\lon V$, $k\geq 1$, of degree $1$ with the property that, for any homogeneous elements $x_1,\cdots,x_{n-1}\in \g,~v\in V$, we have
\begin{eqnarray*}\label{sh-Lie-rep}
&&\sum_{i=1}^{n-1}\sum_{\sigma\in \mathbb S_{(i,n-i-1)} }\varepsilon(\sigma)\rho_{n-i+1}(l_i(x_{\sigma(1)},\cdots,x_{\sigma(i)}),x_{\sigma(i+1)},\cdots,x_{\sigma(n-1)},v)\\
\nonumber&&+\sum_{i=1}^{n}\sum_{\sigma\in \mathbb S_{(n-i,i-1)} }\varepsilon(\sigma)(-1)^{x_{\sigma(1)}+\cdots+x_{\sigma(n-i)}}\rho_{n-i+1}(x_{\sigma(1)},\cdots,x_{\sigma(n-i)},\rho_i(x_{\sigma(n-i+1)},\cdots,x_{\sigma(n-1)},v))=0.
\end{eqnarray*}
\end{defi}

Let $(V,\{\rho_k\}_{k=1}^{+\infty})$ be a representation of an $L_\infty$-algebra $(\g,\{l_k\}_{k=1}^{+\infty})$. There is an $L_\infty$-algebra structure on the direct sum $\g\oplus V$ given by
\begin{eqnarray*}
l_k\big((x_1,v_1),\cdots,(x_k,v_k)\big):=\big(l_k(x_1,\cdots,x_k),\sum_{i=1}^{k}(-1)^{x_i(x_{i+1}+\cdots+x_k)}\rho_k(x_1,\cdots,x_{i-1},x_{i+1},\cdots,x_k,v_i)\big).
\end{eqnarray*}
This $L_\infty$-algebra is called the
{\em semidirect product} of the $L_\infty$-algebra $(\g,\{l_k\}_{k=1}^{+\infty})$ and $(V,\{\rho_k\}_{k=1}^{+\infty})$, and denoted by $\g\ltimes_{\rho}V$.

\begin{defi}\label{de:homoop}
\begin{enumerate}
\item[\rm(i)] Let $(V,\{\rho_k\}_{k=1}^{+\infty})$ be a representation of an $L_\infty$-algebra $(\g,\{l_k\}_{k=1}^{+\infty})$. A degree $0$ element $T=\sum_{k=1}^{+\infty}T_k\in \Hom(\bar{\Sym}(V),\g)$ with $T_k\in \Hom(\Sym^k(V),\g)$ is called a {\em  homotopy relative \RB operator} on an $L_\infty$-algebra $(\g,\{l_k\}_{k=1}^{+\infty})$ with respect to the representation $(V,\{\rho_k\}_{k=1}^{+\infty})$ if the following equalities hold for all $p\geq 1$ and all homogeneous elements $v_1,\cdots,v_p\in V$,

\begin{eqnarray}
\nonumber
\label{full-homotopy-rota-baxter-o}&&\sum_{k_1+\cdots+k_m=t\atop 1\le t\le p-1}\sum_{\sigma\in \mathbb S_{(k_1,\cdots,k_m,1,p-1-t)}}\frac{\varepsilon(\sigma)}{m!}\cdot\\
\nonumber &&T_{p-t}\Big(\rho_{m+1}\Big(T_{k_1}\big(v_{\sigma(1)},\cdots,v_{\sigma(k_1)}\big),\cdots,T_{k_m}\big(v_{\sigma(k_1+\cdots+k_{m-1}+1)},\cdots,v_{\sigma(t)}\big),v_{\sigma(t+1)}\Big),v_{\sigma(t+2)},\cdots,v_{\sigma(p)}\Big)\\
&=&\sum_{k_1+\cdots+k_n=p}\sum_{\sigma\in \mathbb S_{(k_1,\cdots,k_n)}}\frac{\varepsilon(\sigma)}{n!}l_n\Big(T_{k_1}\big(v_{\sigma(1)},\cdots,v_{\sigma(k_1)}\big),\cdots,T_{k_n}\big(v_{\sigma(k_1+\cdots+k_{n-1}+1)},\cdots,v_{\sigma(p)}\big)\Big).
\nonumber
\end{eqnarray}

\item[\rm(ii)] A {\em homotopy relative \RB Lie algebra} is a triple $\big((\g,\{l_k\}_{k=1}^{+\infty}),\{\rho_k\}_{k=1}^{+\infty},\{T_k\}_{k=1}^{+\infty}\big)$, where $(\g,\{l_k\}_{k=1}^{+\infty})$ is an $L_\infty$-algebra, $(V,\{\rho_k\}_{k=1}^{+\infty})$ is a representation of $\g$ on a graded vector space $V$ and $T=\sum_{k=1}^{+\infty}T_k\in \Hom(\bar{\Sym}(V),\g)$ is a  homotopy relative \RB operator.
\end{enumerate}
\end{defi}

A homotopy relative \RB operator on an  $L_\infty$-algebra is a generalization of an $\huaO$-operator on a Lie $2$-algebra introduced in \cite{Sheng}.

A representation of an $L_\infty$-algebra will give rise to a V-data as well as an $L_\infty$-algebra that characterize homotopy relative \RB operators  as \MC elements.

\begin{pro}
Let $(\g,\{l_k\}_{k=1}^{+\infty})$ be an $L_\infty$-algebra and $(V,\{\rho_k\}_{k=1}^{+\infty})$ a representation of $(\g,\{l_k\}_{k=1}^{+\infty})$. Then the following quadruple forms a V-data:
\begin{itemize}
\item[$\bullet$] the graded Lie algebra $(L,[\cdot,\cdot])$ is given by $(C^*(\g\oplus V,\g\oplus V),[\cdot,\cdot]_{\NR})$;
\item[$\bullet$] the abelian graded Lie subalgebra $H$ is given by $H:=\oplus_{n\in\mathbb Z}\Hom^n(\bar{\Sym}(V),\g);$
\item[$\bullet$] $P:L\lon L$ is the projection onto the subspace $H$;
\item[$\bullet$] $\Delta=\sum_{k=1}^{+\infty}(l_k+\rho_k)$.
\end{itemize}

Consequently, $(H,\{\frkl_k\}_{k=1}^{+\infty})$ is an $L_\infty$-algebra, where $\frkl_k$ is given by
\eqref{V-shla}.
\end{pro}

\begin{proof}
By Theorem \ref{graded-Nijenhuis-Richardson-bracket},   $(C^*(\g\oplus V,\g\oplus V),[\cdot,\cdot]_{\NR})$ is a graded Lie algebra. Moreover, by \eqref{graded-NR}, $\Img P=H$ is an abelian graded Lie subalgebra and $\ker P$ is a graded Lie subalgebra. Since $\Delta=\sum_{k=1}^{+\infty}(l_k+\rho_k)$ is the semidirect product  $L_\infty$-algebra structure on $\g\oplus V$, one has $[\Delta,\Delta]_{\NR}=0$ and $P(\Delta)=0$. Thus $(L,H,P,\Delta)$ is a V-data. Hence by Theorem \ref{thm:db-big-homotopy-lie-algebra}, one obtains the higher derived brackets $\{{\frkl_k}\}_{k=1}^{+\infty}$ on the abelian graded Lie subalgebra $H$.
\end{proof}

\begin{thm}\label{hmotopy-o-operator-homotopy-lie}
With the above notation, a degree $0$ element  $T=\sum_{k=1}^{+\infty}T_k\in \Hom(\bar{\Sym}(V),\g)$ is a homotopy relative \RB operator on $(\g,\{l_k\}_{k=1}^{+\infty})$ with respect to the representation $(V,\{\rho_k\}_{k=1}^{+\infty})$ if and only if $T=\sum_{k=1}^{+\infty}T_k$ is a Maurer-Cartan element of the  $L_\infty$-algebra $(H,\{{\frkl_k}\}_{k=1}^{+\infty})$.
\end{thm}

\begin{proof}
See the proof of \cite[Theorem 5.10]{LST}.
\end{proof}

A homotopy relative \RB operator    naturally gives rise to an $L_\infty$-algebra structure on $V$.

\begin{pro}\label{twist-homotopy-lie}
Let $T=\sum_{k=1}^{+\infty}T_k\in \Hom(\bar{\Sym}(V),\g)$ be a homotopy relative \RB operator on $(\g,\{l_k\}_{k=1}^{+\infty})$ with respect to the representation $(V,\{\rho_k\}_{k=1}^{+\infty})$.
\begin{itemize}
  \item[\rm(i)] $e^{[\cdot,T]_\NR}\Big(\sum_{k=1}^{+\infty}(l_k+\rho_k)\Big)$ is a Maurer-Cartan element of the graded Lie algebra $(C^*(\g\oplus V,\g\oplus V),[\cdot,\cdot]_{\NR})$;
  \item[\rm(ii)] there is an $L_\infty$-algebra structure on $V$  given by
\begin{eqnarray}
\nonumber&&\frkl_{t+1}(v_1,\cdots,v_{t+1})=\sum_{k_1+\cdots+k_m=t}\sum_{\sigma\in \mathbb S_{(k_1,\cdots,k_m,1)}}\\
\nonumber&&\frac{\varepsilon(\sigma)}{m!}\rho_{m+1}\Big(T_{k_1}\big(v_{\sigma(1)},\cdots,v_{\sigma(k_1)}\big),\cdots,T_{k_m}\big(v_{\sigma(k_1+\cdots+k_{m-1}+1)},\cdots,v_{\sigma(t)}\big),v_{\sigma(t+1)}\Big);
\end{eqnarray}
  \item[\rm(iii)]  $T$ is an $L_\infty$-algebra homomorphism from the  $L_\infty$-algebra $(V,\{\frkl_k\}_{k=1}^{+\infty})$ to   $(\g,\{l_k\}_{k=1}^{+\infty})$.
\end{itemize}
\end{pro}

\begin{proof}
 See the proof of \cite[Proposition 5.11]{LST}.
\end{proof}

\begin{rmk}
In the classical case, a relative \RB operator induces a pre-Lie algebra \cite{Bai}.  Now a homotopy relative \RB operator also induces a pre-Lie$_\infty$-algebra, which was introduced in \cite{CL}.   See \cite[Section 5.2]{LST} for details.

\end{rmk}

 \begin{rmk}
   Dotsenko and Khoroshkin studied the homotopy of \RB operators on associative algebras in ~\cite{DK} using the operadic approach, and noted that ``in general compact formulas
are yet to be found".  For  \RB  Lie algebras, one encounters a similarly challenging situation. Nevertheless, we use the controlling algebra and Maurer-Cartan approach to give the concrete formulas of homotopy \RB operators, which could provide some guidance for future studies.
 \end{rmk}

 \section{Deformations, cohomologies and homotopies of \RB Lie algebras of nonzero weights}\label{sec:ass}

Note that there is a more general notion of relative \RB Lie algebras of weight $\lambda$, and the relative \RB Lie algebras studied in previous sections are of weight 0. In this section, we briefly review recent developments of deformations, cohomologies and homotopies of \RB Lie algebras of   weight $\lambda$.

Let $(\g, [\cdot,\cdot]_{\g})$  and $ (\h,
[\cdot,\cdot]_{\h})$ be   Lie algebras. Let $\phi: \g \to \Der(\h)$ be a Lie algebra
  homomorphism, which is called an {\em action} of $\g$ on $\h$.

\begin{defi}\label{defi:rb-lie-algebra}
Let  $\phi: \g\rightarrow\Der(\h)$ be an action of a Lie
algebra $(\g, [\cdot,\cdot]_{\g})$ on a Lie algebra $ (\h,
[\cdot,\cdot]_{\h})$. A linear map $T: \h\rightarrow\g$ is called a {\em
  relative \RB operator  of weight $\lambda$}  on $\g$ with respect
to $(\h;\phi)$  if
\begin{equation} \label{eq:B}
[T(u), T(v)]_\g=T\Big(\phi(T(u))v-\phi(T(v))u+\lambda[u,v]_\h\Big),\quad \forall u, v\in\h.
\end{equation}
In particular, if $\g=\h$ and the action    is the adjoint representation of $\g$ on itself, then  $T$ is called a {\em \RB operator of weight $\lambda$}. A {\em \RB Lie algebra of weight $\lambda$} is a Lie algebra equipped with a \RB operator of weight $\lambda$.
\end{defi}

In \cite{TBGS-2}, the notion of a homotopy relative \RB operator of weight $\lambda$ on a symmetric Lie algebra was introduced using the controlling algebra approach, which was the first step toward the definition of a homotopy relative \RB operator of weight $\lambda$ on an $L_\infty$-algebra. As a byproduct, the controlling algebra of relative \RB operators of weight $\lambda$ was given in \cite[Corollary 2.17]{TBGS-2}. More precisely,
let $\phi:\g\lon\Der(\h)$ be an action of a Lie algebra $\g$ on a Lie algebra $\h$. Then $(\oplus_{n=0}^{+\infty}\Hom(\wedge^{n}\h,\g),\Courant{\cdot,\cdot},\dM)$ is a differential graded Lie algebra, where the differential $\dM:\Hom(\wedge^{n}\h,\g)\lon\Hom(\wedge^{n+1}\h,\g)$ is given by
\begin{eqnarray*}
&&(\dM g) ( v_1,\cdots, v_{n+1})\\
&=&\sum_{1\le i< j\le n+1}(-1)^{n+i+j-1}g(\lambda[v_i,v_j]_{\h},v_1,\cdots,\hat{v}_i,\cdots,\hat{v}_{j},\cdots,v_{n+1}),
\end{eqnarray*}
for all $g\in \Hom(\wedge^{n}\h,\g)$ and $v_1,\cdots, v_{n+1} \in\h$, and the graded Lie bracket  $$\Courant{\cdot,\cdot}: \Hom(\wedge^n\h,\g)\times \Hom(\wedge^m\h,\g)\longrightarrow \Hom(\wedge^{m+n}\h,\g)$$ is given by
\begin{eqnarray*}
&&\Courant{g_1,g_2} ( v_1,\cdots, v_{m+n} )\\
&=&\sum_{\sigma\in \mathbb S_{(m,1,n-1)}}(-1)^{1+\sigma}g_1\Big(\phi\big(g_2(v_{\sigma(1)},\cdots,v_{\sigma(m)})\big)v_{\sigma(m+1)}, v_{\sigma(m+2)},\cdots,v_{\sigma(m+n)}\Big)\\
&&\sum_{\sigma\in \mathbb S_{(n,1,m-1)}}(-1)^{mn+\sigma}g_2\Big(\phi\big(g_1(v_{\sigma(1)},\cdots,v_{\sigma(n)})\big)v_{\sigma(n+1)}, v_{\sigma(n+2)},\cdots,v_{\sigma(m+n)}\Big)\\
&&\sum_{\sigma\in \mathbb S_{(n,m)}}(-1)^{1+mn+\sigma}[g_1(v_{\sigma(1)},\cdots,v_{\sigma(n)}),g_2(v_{\sigma(n+1)},\cdots,v_{\sigma(m+n)})]_{\g},
\end{eqnarray*}
for all  $g_1\in \Hom(\wedge^n\h,\g),~g_2\in \Hom(\wedge^m\h,\g)$ and $v_1,\cdots, v_{m+n} \in\h.$
Moreover,
a linear map $T:\h\lon\g$  is a relative \RB operator of weight $\lambda$ on $\g$ with respect to the action $\phi$ if and only if $T$ is a Maurer-Cartan element of the above differential graded Lie algebra.

As soon as one has the above controlling algebra of relative \RB operators of weight $\lambda$, one can obtain immediately the  differential graded Lie algebra that controls deformations of a relative \RB operator $T$ of weight $\lambda$  using the twisted differential $\dM_T:=\dM+\Courant{T,\cdot}$. Meanwhile, one can also define the cohomology of a relative \RB operators $T$ of weight $\lambda$ using the twisted differential $\dM_T.$  See \cite{Das2108} for details. Note that  in \cite{Das2108}, the  controlling algebra of relative \RB operators of weight $\lambda$ on associative algebras were constructed parallelly.

Before \cite{Das2108}, the cohomologies of relative \RB operators of weight $1$ on Lie algebras were given in \cite{JSZ} using a different approach. Namely a   relative \RB operator $T:\h\to \g$ of weight $1$ induces a new Lie algebra $(\h,[\cdot,\cdot]_T)$ and a representation  $\theta: \h\rightarrow\gl(\g)$ of $(\h,[\cdot,\cdot]_T)$ on the vector space $\g$, where $[\cdot,\cdot]_T$ and  $\theta$ are given  by
\begin{eqnarray}
\label{eq:desLieb}[u,v]_T&=&\phi(T(u))v-\phi(T(v))u+[u,v]_{\h}, \\
\label{lierepT}
\theta(u)x&=&T(\phi(x)u)+[T(u), x]_\g.
\end{eqnarray}
 The Chevalley-Eilenberg cohomology of the Lie algebra $(\h,[\cdot,\cdot]_T)$ with coefficients in the representation $\theta$ is taken to be the cohomology of the relative \RB operator $T$. In the same paper, the cohomologies of relative \RB operators of weight $1$ on Lie groups were also introduced and the classical Van Est map was extended to the context of cohomologies of relative \RB operators   on Lie groups and Lie algebras.

 Using Voronov's higher derived brackets \cite{Vo},  Caseiro and   Nunes da Costa succeed in defining homotopy relative \RB operators of weight 1 on $L_\infty$-algebras with respect to $L_\infty$-actions \cite{CC}, generalized some results in \cite{LST} and \cite{TBGS-2}. In the associative algebra context,
  Wang and Zhou studied  homotopy \RB associative algebras of weight $\lambda$  in \cite{WZ} and showed that the operad governing
homotopy \RB associative  algebras is a minimal model of the operad of \RB associative algebras. The cohomologies of \RB associative algebras of weight $\lambda$ were also given in the same paper. Parallelly, the cohomologies of \RB Lie algebras of weight $\lambda$ were   given in \cite{Das2109} by which   abelian extensions and formal  deformations are studied.

\vspace{2mm}

\noindent{\em Acknowledgements.} This research was partially supported by NSFC (11922110).


\begin{thebibliography}{a}

 \bibitem{Arnal}
 D. Arnal,  Simultaneous deformations of a Lie algebra and its modules. Differential geometry and mathematical physics (Liege, 1980/Leuven, 1981), 3-15, \emph{Math. Phys. Stud.}, 3, Reidel, Dordrecht, 1983.



\bibitem{Bai}
C. Bai, A unified algebraic approach to the classical Yang-Baxter equation. \emph{J. Phys. A: Math. Theor.} {\bf 40} (2007), 11073-11082.

\bibitem{BBGN} C. Bai, O. Bellier, L. Guo and X. Ni, Spliting of operations, Manin products and Rota-Baxter operators.  {\em Int. Math. Res. Not.} {\bf 3} (2013), 485-524.

\bibitem{Bal} D. Balavoine, Deformations of algebras over a quadratic operad. Operads: Proc. of Renaissance Conferences (Hartford, CT/Luminy, 1995), \emph{Contemp. Math.} {\em 202} Amer. Math. Soc., Providence, RI, 1997, 207-34.

\bibitem{Barmeier}
S. Barmeier and Y. Fr\'egier, Deformation-obstruction theory for diagrams of algebras and applications to geometry.  {\em J. Noncommut. Geom.}    {\bf14} (2020), no. 3, 1019-1047.


\bibitem{Ba} G. Baxter, An analytic problem whose solution follows from a simple algebraic identity. \emph{Pacific J. Math.} {\bf 10} (1960), 731-742.





\bibitem{Borisov}
D. V. Borisov, Formal deformations of morphisms of associative algebras. {\em Int. Math. Res. Not.} {\bf 41} (2005), 2499-2523.




\bibitem{CC}
R. Caseiro and J. Nunes da Costa, O-Operators on Lie $\infty$-algebras with respect to Lie $\infty$-actions. \emph{Comm. Algebra} {\bf50} (7) (2022), 3079-3101.


\bibitem{CL}
F. Chapoton and M. Livernet, Pre-Lie algebras and the rooted trees operad. {\em  Int. Math. Res. Not.}  {\bf 8} (2001), 395-408.


\bibitem{Ch-Ei}
C. Chevalley and S. Eilenberg,
\newblock Cohomology theory of {L}ie groups and {L}ie algebras.
\newblock {\em Trans. Amer. Math. Soc.} {\bf 63} (1948), 85-124.

\bibitem{CK}
A. Connes and D. Kreimer, { Renormalization in quantum field theory and the Riemann-Hilbert problem. I. The Hopf algebra structure of graphs and the main theorem.} {\em Comm. Math. Phys.} {\bf 210} (2000), 249-273.

    \bibitem{Das}
A. Das, Deformations of associative Rota-Baxter operators. \emph{J. Algebra} {\bf560} (2020), 144-180.

\bibitem{DasM}
A. Das and S. Mishra, The $L_\infty$-deformations of associative Rota-Baxter algebras and homotopy Rota-Baxter operators. \emph{J. Math. Phys.} {\bf63} (2022), 051703.

\bibitem{Das2108}
A. Das, Cohomology and deformations of weighted Rota-Baxter operators. \emph{J. Math. Phys.} 63 (2022), 091703.

\bibitem{Das2109}
A. Das, Cohomology of weighted Rota-Baxter Lie algebras and Rota-Baxter paired operators. arXiv:2109.01972.

\bibitem{Dolgushev-Rogers}
V. A. Dolgushev and C. L. Rogers, A version of the Goldman-Millson Theorem
for filtered $L_\infty$-algebras. {\em J. Algebra}  {\bf 430} (2015), 260-302.


\bibitem{DK}
V. Dotsenko and A. Khoroshkin, Quillen homology for operads via Gr\"{o}bner bases. \emph{Doc. Math.} {\bf 18} (2013), 707-747.


\bibitem{Fard}
K Ebrahimi-Fard, D. Manchon and F. Patras, A noncommutative Bohnenblust-Spitzer identity for Rota-Baxter algebras solves Bogoliubov's counterterm recursion.  {\em J. Noncommut. Geom.} {\bf 3} (2009), 181-222.




\bibitem{Fregier}
Y. Fr\'egier, M. Markl and D. Yau,   The
$L_\infty$-deformation complex of diagrams of algebras. \emph{New York J. Math.} {\bf 15} (2009), 353-392.

\bibitem{Fregier-Zambon-1}
Y. Fr\'egier and M. Zambon, Simultaneous deformations and Poisson geometry. \emph{Compos. Math.} {\bf 151} (2015), 1763-1790.

\bibitem{Fregier-Zambon-2}
Y. Fr\'egier and M. Zambon, Simultaneous deformations of algebras and morphisms via derived brackets. \emph{J. Pure Appl. Algebra} {\bf 219 } (2015), 5344-5362.



\bibitem{Ge0}
M. Gerstenhaber, The cohomology structure of an associative ring. \emph{Ann. Math.} {\bf 78} (1963), 267-288.

\bibitem{Ge}
M. Gerstenhaber, On the deformation of rings and algebras. \emph{Ann. Math. (2) } {\bf 79} (1964), 59-103.




\bibitem{Getzler}
E. Getzler, Lie theory for nilpotent $L_{\infty}$-algebras. {\em Ann.   Math. (2)} {\bf 170} (2009), 271-301.

\bibitem{Goncharov}
M. E. Goncharov and P. S. Kolesnikov, Simple finite-dimensional double algebras. {\em J. Algebra}  {\bf 500} (2018), 425-438.

\bibitem{GLST1} A. Guan, A. Lazarev, Y. Sheng and R. Tang, {Review of deformation theory I: Concrete formulas for deformations of algebraic structures}. \emph{Adv. Math. (China)} {\bf49}  (2020), 257-277.

\bibitem{GLST} A. Guan, A. Lazarev, Y. Sheng, and R. Tang, Review of deformation theory II: a homotopical approach.  	\emph{Adv. Math. (China)} {\bf 49} (2020), 278-298.

\bibitem{Gub-AMS}
L. Guo,  What is a Rota-Baxter algebra? {\em Notices of the AMS} {\bf 56}  (2009), 1436-1437.



\bibitem{Gub}
L. Guo,  An introduction to Rota-Baxter algebra. Surveys of Modern Mathematics, 4. International Press, Somerville, MA; Higher Education Press, Beijing, 2012. xii+226 pp.




\bibitem{Hamilton-Lazarev}
A. Hamilton and A. Lazarev, Cohomology theories for homotopy algebras and noncommutative geometry. \emph{Algebr. Geom. Topol.} {\bf 9} (2009), 1503-1583.

\bibitem{Har}
D.~K. Harrison,
\newblock Commutative algebras and cohomology.
\newblock {\em Trans. Amer. Math. Soc.}  {\bf 104} (1962), 191-204.


\bibitem{Hor}
G.~Hochschild,
\newblock On the cohomology groups of an associative algebra.
\newblock {\em Ann.  Math. (2)}  {\bf 46} (1945), 58-67.

\bibitem{JSZ} J. Jiang, Y. Sheng and C. Zhu, Lie theory and cohomology of relative Rota-Baxter operators.  arXiv:2108.02627.

\bibitem{KSo}
M. Kontsevich and Y. Soibelman, Deformation theory. I [Draft], http://www.math.ksu.edu/~soibel/Book-vol1.ps, 2010.

\bibitem{Kosmann-Schwarzbach}
 Y. Kosmann-Schwarzbach, From Poisson algebras to Gerstenhaber algebras. {\em Ann. Inst. Fourier (Grenoble) } {\bf 46} (1996), 1243-1274.


\bibitem{Ku}
B. A. Kupershmidt, What a classical $r$-matrix really is. \emph{J.
Nonlinear Math. Phys.} {\bf 6} (1999), 448-488.

\bibitem{LS}
T. Lada and J. Stasheff, Introduction to sh Lie algebras for
physicists. \emph{Internat. J. Theoret. Phys.} {\bf 32} (1993), 1087-1103.

\bibitem{LM}
T. Lada and M. Markl,  Strongly homotopy Lie algebras. \emph{ Comm. Algebra} {\bf 23} (1995),  2147-2161.



\bibitem{LST}
A. Lazarev, Y. Sheng and R. Tang, Deformations and homotopy theory  of relative Rota-Baxter Lie algebras.  \emph{Comm. Math. Phys.} {\bf 383} (2021),  595-631.

\bibitem{LST2}
A. Lazarev, Y. Sheng and R. Tang,  Homotopy relative Rota-Baxter Lie algebras, triangular $L_\infty$-bialgebras and higher derived brackets. arXiv:2008.00059. \emph{Tran. Amer. Math. Soc.} (2023) https://doi.org/10.1090/tran/8844.


\bibitem{LV} J.-L. Loday and B. Vallette, Algebraic Operads. Springer, 2012.

\bibitem{Lu} J. Lurie, \emph{ DAG X: Formal moduli problems}, available at http://www.math.harvard.edu/~lurie/papers/DAG-X.pdf

\bibitem{Ma-0}
M. Markl, Intrinsic brackets and the $L_{\infty}$-deformation theory of bialgebras. \emph{J. Homotopy Relat. Struct.} {\bf5} (2010), 177-212.


\bibitem{Ma}
M. Markl, Deformation Theory of Algebras and Their Diagrams. \emph{Regional Conference Series in Mathematics}, Number 116, American Mathematical Society (2011).

\bibitem{MSS} M. Markl, S. Shnider and J. D. Stasheff, Operads in
    Algebra, Topology and Physics. American Mathematical Society, Providence, RI, 2002.





\bibitem{NR} A. Nijenhuis  and R. Richardson,  Cohomology and deformations in graded Lie algebras. {\em Bull.
Amer. Math. Soc.} {\bf 72} (1966), 1-29.

\bibitem{NR2} A. Nijenhuis and R. Richardson,  Commutative  algebra cohomology and deformations of Lie and associative algebras. {\em J. Algebra} {\bf 9} (1968), 42-105.

\bibitem{PBG} J. Pei, C. Bai and L. Guo, Splitting of Operads and
Rota-Baxter Operators on Operads. \emph{Appl. Categor. Struct.}
{\bf 25} (2017), 505-538.

\bibitem{Pr} J. P. Pridham, Unifying derived deformation theories. \emph{Adv. Math.} {\bf 224} (2010), 772-826.



\bibitem{STS} M.~A. Semyonov-Tian-Shansky, What is a
classical R-matrix? \emph{Funct. Anal. Appl.} {\bf 17} (1983), 259-272.




\bibitem{Sheng}
Y. Sheng, Categorification of pre-Lie algebras and solutions of $2$-graded classical Yang-Baxter equations. \emph{Theory Appl.  Categ.} {\bf 34} (2019), 269-294.

\bibitem{Sta63} J. Stasheff, Homotopy associativity of H-spaces. I, II. \emph{Trans. Amer. Math. Soc.} {\bf 108} (1963), 275-292; ibid. {\bf 108} (1963), 293-312.


\bibitem{stasheff:shla} J. Stasheff,  Differential graded {L}ie algebras, quasi-Hopf algebras and higher homotopy algebras.   \emph{Quantum groups (Leningrad, 1990)}, 120-137, Lecture Notes in Math., 1510, \emph{ Springer, Berlin,} 1992.



\bibitem{TBGS-1}
R. Tang, C. Bai, L. Guo and Y. Sheng, Deformations and their controlling cohomologies of $\huaO$-operators. \emph{Comm. Math. Phys.} {\bf 368} (2019), 665-700.

\bibitem{TBGS-s}
R. Tang, C. Bai, L. Guo and Y. Sheng,   Deformation and homotopy of Rota-Baxter operators and $\huaO$-operators on Lie algebras. \emph{Phys. Part. Nuclei} {\bf51} (4) (2020), 393-398.

\bibitem{TBGS-2}
R. Tang, C. Bai, L. Guo and Y. Sheng, Homotopy Rota-Baxter operators, homotopy $\huaO$-operators and homotopy post-Lie algebras. {arXiv:1907.13504}, to appear in \emph{J. Noncommut. Geom.}


\bibitem{Vo}
Th. Voronov, Higher derived brackets and homotopy algebras. \emph{J. Pure Appl. Algebra}  {\bf 202} (2005), 133-153.

\bibitem{WZ}
K. Wang and G. Zhou, Deformations and homotopy theory of Rota-Baxter algebras of any weight. arXiv:2108.06744.


\bibitem{Yu-Guo}
H. Yu, L. Guo and J.-Y. Thibon, Weak quasi-symmetric functions, Rota-Baxter algebras and Hopf algebras. \emph{Adv. Math.} {\bf 344} (2019), 1-34.



\end{thebibliography}
\end{document}